\documentclass[american,aps,pra,reprint, superscriptaddress, showpacs]{revtex4-2}
\usepackage[unicode=true,pdfusetitle, bookmarks=true,bookmarksnumbered=false,bookmarksopen=false, breaklinks=false,pdfborder={0 0 0},backref=false,colorlinks=false]{hyperref}
\hypersetup{colorlinks,linkcolor=myurlcolor,citecolor=myurlcolor,urlcolor=myurlcolor}
\usepackage{graphics,epstopdf,graphicx, amsthm, amsmath, amssymb, times, braket, color, bm, physics,comment,bbm}

\definecolor{myurlcolor}{rgb}{0,0,0.7}

\newcommand{\iden}{\mathbb{I}}

\newtheorem{lemma}{Lemma}

\newtheorem{conjecture}{Conjecture}

\begin{document}


\title{Exploring the gap between thermal operations and enhanced thermal operations}

\author{Yuqiang Ding}
\author{Feng Ding}
\author{Xueyuan Hu}%
 \email{xyhu@sdu.edu.cn}
\affiliation{
 School of Information Science and Engineering, Shandong University, Qingdao 266237, China
}%




\date{\today}

\begin{abstract}
The gap between thermal operations (TO) and enhanced thermal operations (EnTO) is an open problem raised in [Phys. Rev. Lett. 115, 210403 (2015)]. It originates from the limitations on coherence evolutions. Here we solve this problem by analytically proving that, a state transition induced by EnTO cannot be approximately realized by TO. It confirms that TO and EnTO lead to different laws of state conversions. Our results can also contribute to the study of the restrictions on coherence dynamics in quantum thermodynamics.
\end{abstract}

\maketitle


\section{\label{sec:level1}Introduction}
In the resource theory of quantum thermodynamics \cite{Lostaglio_2019}, the main problem is to figure out the allowed state conversions under a set of quantum operations known as thermal operations (TO). A thermal operation can be constructed as follows \cite{Janzing2000a,Horodecki2013,PhysRevLett.111.250404}. A quantum system, previously isolated and characterized by a Hamiltonian $H_S$, is brought into contact with a heat bath described by a Hamiltonian $H_R$ at a fixed inverse temperature $\beta$. Then the system is decoupled from the bath after some time. This interaction conserves energy overall through the whole process, according to the first law of thermodynamics.
 
Although the definition of TO has clear operational meaning, it is difficult to be dealt with mathematically, because the number of variables for describing the interaction is infinitely large. Alternatively, two properties of TO are observed: (1) the time-translation symmetry, which corresponds to the first law, and (2) the Gibbs-preserving condition, which corresponds to the second law. The set of quantum operations which satisfy both of these properties are called the enhanced thermal operations (EnTO) \cite{PhysRevLett.115.210403,Gour2018}. When only population dynamics is concerned, it has been proven that state conversions induced by TO are equivalent to those induced by EnTO \cite{Horodecki2013}. This elegant result leads to the necessary and sufficient conditions on population dynamics under TO \cite{PhysRevLett.111.250404,Horodecki2013,brandao2015second,PhysRevX.8.041051}.

In quantum systems, coherence between energy levels cannot be created or enhanced by thermal operations, and hence are widely studied as an independent resource aside from non-equilibrium populations \cite{PhysRevLett.115.210403,PhysRevX.5.021001,Narasimhachar2015,Lostaglio2015,Korzekwa_2016,PhysRevA.96.032109,PhysRevA.99.012104}. In Refs. \cite{PhysRevLett.115.210403,PhysRevX.5.021001}, the time-transitional symmetric property is employed to derive an upper bound for the coherence preserved by TO. This bound is tight for qubit systems, but not for high-dimensional systems \cite{PhysRevLett.115.210403}. This leads to a gap between TO and EnTO, namely, there are state conversions under EnTO which cannot be realized \emph{exactly} by TO. However, it remains an open problem whether this gap can be closed approximately. Here we formally state two versions of the closure conjecture.

\begin{conjecture}\label{con1}
(Closure conjecture, v1) For any enhanced thermal operation $\mathcal{E}^{\rm EnTO}$, there exists a thermal operation $\mathcal{E}^{\rm TO}$, such that the distance between these two operations is small, i.e., $|\mathcal{E}^{\rm EnTO}-\mathcal{E}^{\rm TO}|\leq\epsilon$.
\end{conjecture}

\begin{conjecture}\label{con2}
(Closure conjecture, v2) For any given input state $\rho_0$, if the state conversion $\rho_0\rightarrow\rho$ is realizable by EnTO, then there always exist a state $\rho'$ such that $|\rho'-\rho|\leq\epsilon$ and the conversion $\rho_0\rightarrow\rho'$ is achievable by TO.
\end{conjecture}

Apparently, Conjecture \ref{con1} implies Conjecture \ref{con2}. Thus disproof of Conjecture \ref{con2} is sufficient for disproving Conjecture \ref{con1}. In this paper, we will disprove Conjecture \ref{con2} with an analytical counterexample, and hence confirms that TO and EnTO lead to different laws of state conversions. Precisely, we consider a qutrit system with Hamiltonian $H_S=\sum_{m=0}^2m\hbar\omega\ket{m}\bra{m}$ and in initial state $\ket{\psi_0}=\frac{1}{\sqrt{2}}(\ket{0}+\ket{1})$, and a heat bath at a fixed inverse temperature $\beta$. We will show that the state conversion from $\rho_0=\ket{\psi_0}\bra{\psi_0}$ to 
\begin{equation}
    \rho=\frac12\left(\begin{array}{ccc}
      1-e^{-2\beta\hbar\omega}   & \sqrt{1-e^{-2\beta\hbar\omega}} & 0 \\
      \sqrt{1-e^{-2\beta\hbar\omega}}   & 1 & 0\\
      0 & 0 & e^{-2\beta\hbar\omega}
    \end{array}
    \right)
\end{equation}
is realizable by EnTO. Then we will prove analytically that, if the temperature is proper such that $0\ll e^{-\beta\hbar\omega}\ll \frac{\sqrt{5}-1}{2}$, then any state $\rho'$ satisfying $|\rho'-\rho|\leq\epsilon$ is not achievable by TO.

\section{Thermal operations and related concepts}
 
	Here we brief review some related concepts and results.
	By definition, a thermal operation is expressed as \cite{Horodecki2013}
	\begin{equation}
	 \mathcal{E}^{\rm TO}(\rho_{S})= {\rm Tr}_{R}[{U}(\rho_{S}\otimes\gamma_{R}){U}^{\dagger}],
	\label{1}
	\end{equation}
	where $\rho_S$ is a quantum state of the system with Hamiltonian $H_S$, $\gamma_{R}=e^{-\beta H_R}/{\rm Tr}(e^{-\beta H_R})$ is the Gibbs state of the heat bath with Hamiltonian $H_R$ at a fixed inverse temperature $\beta$, and $U$ is a joint unitary commuting with the total Hamiltonian of the system and heat bath, $[U,H_S+H_R]=0$.
 
	In Ref. \cite{PhysRevLett.115.210403}, the following two core properties of TO are identified:
	\newline
   	(i) Time-translation symmetric condition,
	\begin{equation}	
 \mathcal{E}^{\rm TO}(e^{-iH_St}\rho_Se^{iH_St})=e^{-iH_St} \mathcal{E}^{\rm TO}(\rho_S)e^{iH_St};
\label{2}
	\end{equation}
	(ii) Gibbs-preserving condition,
	\begin{equation}
 \mathcal{E}^{\rm TO}(\gamma_{S})=\gamma_S.
	\end{equation}
  
These properties are related to the laws of thermodynamics. Property (i) is derived from the energy conservation condition, and thus reflects the first law.
Property (ii) describes the second law, i.e, it is impossible to prepare a non-equilibrium state from an equilibrium state without consuming extra work. 
The operations satisfying both properties (i) and (ii) are called enhanced thermal operations (EnTO) \cite{PhysRevX.5.021001} or thermal processes \cite{Gour2018}.
	
Let $\boldsymbol{p_0}$ and $\boldsymbol{p}$ be the vectors of population distributions of states $\rho_0$ and $\rho$, respectively. Each element of the vector $\boldsymbol{p}$ is  $p_k=\langle k\vert\rho\vert k\rangle$ (and similar for $\boldsymbol{p_0}$), where $\ket{k}$ are energetic eigenstates of the system. The population dynamics induced by an enhenced thermal operation $\mathcal{E}^{\rm EnTO}$ can be written as
\begin{equation}
	\rho= \mathcal{E}^{\rm EnTO}(\rho_0)\ \Rightarrow \  \boldsymbol{p}=G \boldsymbol{p_0}.\label{eq:popu}
\end{equation}
Here $G$ is a matrix of transition probabilities $G_{k'k}=p_{k'|k}=\langle k'\vert \mathcal{E}(\vert k\rangle\langle k\vert)\vert k'\rangle$ from state $\ket{k}$ to $\ket{k'}$. From property (ii), the population dynamics $G$ induced by EnTO is a stochastic matrix that preserve the Gibbs distribution. Such matrices, also called the Gibbs-stochastic matrices, can be realized by TO  \cite{Horodecki2013}. Hence, when only population dynamics is concerned, TO is equivalent to EnTO.
 
As shown in Refs.  \cite{PhysRevLett.115.210403,PhysRevX.5.021001}, the coherence dynamics between energy levels depends on both initial coherence of quantum state and transition probabilities. For a quantum state $\rho$ expanded in its energy eigenbasis $\rho=\sum_{i,j}\rho_{ij}\vert i\rangle\langle j\vert$, a mode of coherence is defined as an operator $\rho^{(\omega)}$ composed of coherence terms between degenerate gaps
	\begin{equation}
	\rho^{(\omega)}=\sum_{i,j:E_i-E_j=\hbar\omega}\rho_{ij}\vert i\rangle\langle j\vert.
	\end{equation}
	The output coherence term after the action of an enhanced thermal operation is bounded as \cite{PhysRevLett.115.210403,PhysRevX.5.021001}
	\begin{equation}
	\vert\rho_{ij}\vert\leq{\sum_{c,d}}'\vert\rho_{0,cd}\vert\sqrt{p_{i\vert c}p_{j\vert d}},\label{7}
	\end{equation}	
	where the primed sum $\sum'$ refers to the summation over the indices $c, d$ which satisfy $E_c-E_d=E_i-E_j$, and $\rho_{0,cd}=\bra{c}\rho_0\ket{d}$. 
 
If $\rho_0$ and $\rho$ are qubit states, TO can saturate the bound in Eq. (\ref{7}). However, for higher dimension systems, there exist situations where TO cannot achieve the bound in Eq. (\ref{7}) exactly, while EnTO can \cite{PhysRevLett.115.210403}.
	
\section{Setup}
In this paper, the main system we focus on is a three-dimensional quantum system whose Hamiltonian reads
\begin{equation}
	H_S=\sum_{m=0}^{2}m\hbar\omega\vert m\rangle\langle m\vert.
	\label{8}
\end{equation}
The heat bath is at a fixed inverse temperature $\beta$. Here and following, we will label $q\equiv e^{-\beta\hbar\omega}$.
	
Because the main purpose of this paper is to disprove the closure conjecture, a counterexample is sufficient. Thus in the rest of the paper, we mainly consider the initial state $\ket{\psi_0}=\frac{1}{\sqrt{2}}(\ket{0}+\ket{1})$, though most of our discussions can be applied to general input states.

The set of states which can be obtained from a given input state $\rho$ via a set of operations $X$ is called the $X$ cone of $\rho$, labeled as $\mathcal{C}^X(\rho):=\{\rho':\rho'=\mathcal{E}(\rho),\mathcal{E}\in X\}$. The gap between two sets of operations $X_1$ and $X_2$ can be indicated from a gap between $X_1$ and $X_2$ cone of a given state. In the following, we explicitly calculate the EnTO cone of the state $\rho_0=\ket{\psi_0}\bra{\psi_0}$, and show that the gap between TO and EnTO cones of $\rho_0$ is non-negligible.

From Eqs. (\ref{eq:popu}) and (\ref{7}), any state in the EnTO (or TO) cone of $\rho_0$ is in the following form
\begin{equation}
    \rho=\left(
    \begin{array}{ccc}
        p_0 & |\rho_{10}|e^{i\phi_1} & 0 \\
        |\rho_{10}|e^{-i\phi_1} & p_1 & |\rho_{21}|e^{i\phi_2} \\
        0 & |\rho_{21}|e^{i\phi_2} & 1-p_0-p_1
    \end{array}
    \right).
\end{equation}
Further, all states in the above form are equivalent to states with $\phi=0$ by covariant unitary operators $U(\phi)=\mathrm{diag}(e^{-i\phi_1},1,e^{i\phi_2})$. Therefore, the EnTO (or TO) cone of $\rho_0$ is fully described by $(p_0,p_1,|\rho_{10}|,|\rho_{21}|)$, and hence can be presented in a four-dimensional parameter space. Here and following, we will focus on the projections of the cones on the tree-dimention parameter space $(p_0,p_1,|\rho_{10}|)$. Clearly, a gap between the projections of cones indicates a gap between the cones.

\section{the gap between TO and EnTO}	
\subsection{EnTO cone}
Here we analytically calculate the EnTO cone of the initial state $\rho_0$, and visualize it in the three-dimensional parameter space spanned by $(p_0,p_1,\vert\rho_{10}\vert)$, see Fig. \ref{fig1}. The range of the output populations is derived from the thermo-majorization relation \cite{Horodecki2013}.

	\begin{figure}
		\includegraphics[scale=0.85]{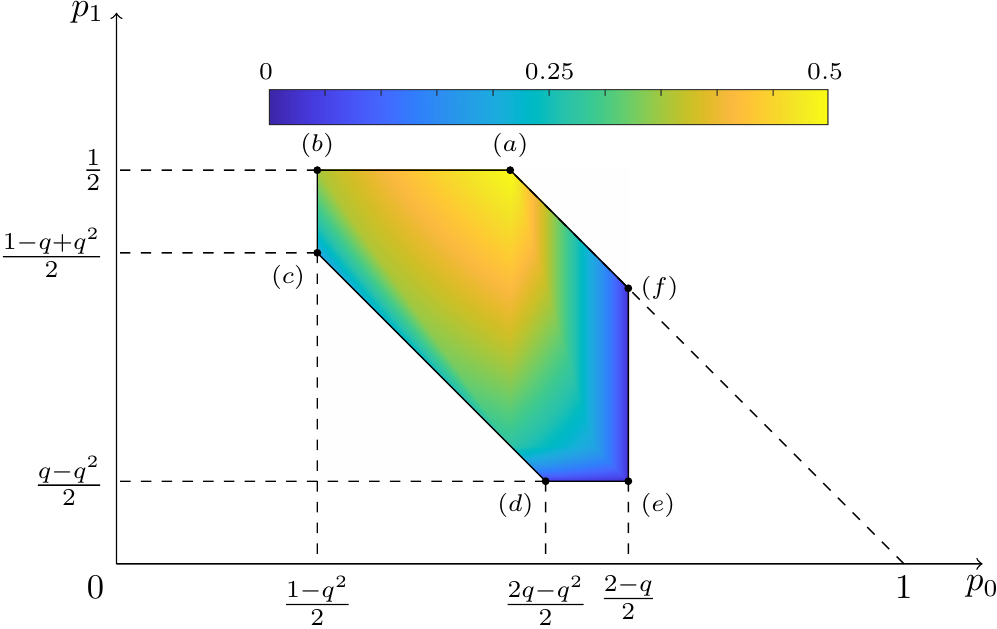}
		\centering
		\caption{EnTO cone of the qutrit state $\ket{\psi_0}=\frac{1}{\sqrt{2}}(\ket{0}+\ket{1})$. The maximum value of $\vert\rho_{01}\vert$ is presented as the color bar. Here the parameter $q\equiv e^{-\beta\hbar\omega}$.}
		\label{fig1}
	\end{figure}

The maximum value of $\vert\rho_{10}\vert$ for given $(p_0,p_1)$ is obtained as follows. According to Eq. (\ref{7}), we have $\vert\rho_{10}\vert\leq\frac12\sqrt{G_{00}G_{11}}$. Thus the upper bound of $\vert\rho_{10}\vert$ reads
\begin{eqnarray}
&& \max \frac12\sqrt{G_{00}G_{11}} \nonumber\\
&s.t. \ \  & G\boldsymbol{\gamma}=\boldsymbol{\gamma}, G\boldsymbol{p_0}=\boldsymbol{p}, G^{\mathrm{T}}\mathbb{I}=\mathbb{I}, G_{ij}\in[0,1].\label{eq:max_coh}
\end{eqnarray}
where $\boldsymbol{\gamma}=(1,q,q^2)^\mathrm{T}$, $\boldsymbol{p_0}=(\frac12,\frac12,0)^\mathrm{T}$, $\boldsymbol{p}=(p_0,p_1,1-p_0-p_1)^\mathrm{T}$, and $\mathbb{I}=(1,1,1)^\mathrm{T}$.
This bound is reached by the EnTO with Kraus operators
\begin{equation}
\label{eq:kraus} K^{(n)}=\sum_{i,j=0:i-j=n}^{2}\sqrt{G^\star_{ij}}\vert i\rangle\langle j\vert,
\end{equation}
where $n=-2,-1,0,1,2$ and $G^\star_{ij}$ denotes the elements of the optimal transition matrix $G^\star$ which reach the maximum in Eq. (\ref{eq:max_coh}).
We analytically solve the problem in Eq. (\ref{eq:max_coh}) (see Appendix \ref{app:ento_cone} for details) and presented the result in Fig. \ref{fig1}. As shown in this figure, the extreme states of the EnTO cone is continuous, namely, a small perturbation in the output population would only result in a small variance of the maximum value of $\vert\rho_{10}\vert$.

\subsection{Optimal output coherence via TO}
In order to derive the maximal value of $|\rho_{10}|$ via TO, we start from the general form of the Hamiltonian of the heat bath
\begin{equation}
	H_R=\sum_{E_R}E_R\Pi_{E_R},
\end{equation}
where $E_R$ are the eigenvalues of energy, and $\Pi_{E_R}$ is the projector to the eigenspace of $E_R$. Here we assume that the degeneracy of $E_R$ is monotonically non-decreasing with $E_R$.

Without loss of generality, each eigenvalue $E_R$ can be expressed as
$E_R=n\hbar\omega+\xi\equiv E(\xi,n)$, where $n$ is a non-negative integer and $\xi\in[0,\hbar\omega)$.
Accordingly, $\Pi_{E_R}\equiv\Pi_{\xi,n}$.
Now we divide the Hilbert space of $R$ as $\mathcal{H}_R=\bigoplus_{\xi}\mathcal{H}_\xi$ with $\mathcal{H}_\xi=\bigoplus_{n}\mathcal{H}_{\xi,n}$ and $\mathcal{H}_{\xi,n}$ the eigenspace of $E(\xi,n)$.
The Hamiltonian of the heat bath is then rewritten as $H_R=\bigoplus_{\xi}H_{\xi}$, where $H_{\xi}=\sum_{n}(n\hbar\omega+\xi)\Pi_{\xi,n}$ is the Hamiltonian acting on $\mathcal H_\xi$. It follows that the thermal state of $R$ reads 
\begin{equation}
    \gamma_R=\bigoplus_{\xi}p_\xi\gamma_{\xi},\label{eq:gamma}
\end{equation}
where $\gamma_\xi=e^{-\beta H_{\xi}}/\mathrm{Tr}\left(e^{-\beta H_{\xi}}\right)$ is the thermal state of $H_{\xi}$, $p_\xi=\mathrm{Tr}(e^{-\beta H_{\xi}})/[\sum_{\xi'}\mathrm{Tr}(e^{-\beta H_{\xi'}})]$ satisfies $\sum_\xi p_\xi=1$.

Further, the Hamiltonian of the total system reads
\begin{equation}
	H_{SR}=\bigoplus_{\xi}H_{SR_\xi},
\end{equation}
where $H_{SR_\xi}=H_S\otimes\mathbb{I}_{\xi}+\mathbb{I}_{S}\otimes H_{\xi}$ is the Hamiltonian acting on subspace $\mathcal{H}_S\otimes\mathcal H_{\xi}$. Importantly, for $H_S=\sum_m m\hbar\omega\vert m\rangle\langle m\vert$, the eigenvalues of $H_{SR_\xi}$ and $H_{SR_{\xi'}}$ (where $\xi'\neq \xi$) does not have an overlap. Therefore, any joint unitary which satisfies $[U,H_{SR}]=0$ is in a block-diagonal form
\begin{equation}
    U=\bigoplus_{\xi}U_\xi,\label{eq:u}
\end{equation}
where $U_\xi$ acts on $\mathcal{H}_S\otimes\mathcal H_{\xi}$ and satisfies $[U_\xi,H_{SR_\xi}]=0$. Now we introduce a lemma which will simplify our subsequent analysis.
	 
\begin{lemma} \label{lemma1}
For a system with Hamiltonian $H_S=\sum_m m\hbar\omega\vert m\rangle\langle m\vert$, any thermal operation can be written as $\mathcal{E}^{\rm TO}=\sum_\xi p_\xi \mathcal{E}_{\xi}$, where $\mathcal{E}_{\xi}$ is a thermal operation induced by a heat bath with Hamiltonian $H_{\xi}=\sum_{n}(n\hbar\omega+\xi)\Pi_{\xi,n}$.
\end{lemma}
	 
\begin{proof} 
From Eqs. (\ref{eq:gamma}) and (\ref{eq:u}), a thermal operation acting on $\rho_S$ reads
	\begin{equation}
	\begin{split}
	\mathcal{E}^{\rm TO}(\rho_S)&={\rm Tr}_R[U(\rho_{S}\otimes\gamma_{R})U^{\dagger}]\\
	&={\rm Tr}_{R}(\bigoplus_{\xi}U_\xi(\rho_S\otimes\bigoplus_{\xi}p_\xi\gamma_\xi)\bigoplus_{\xi}U_\xi^{\dagger})\\
	&=\sum_{\xi}p_\xi{\rm Tr}_{R_\xi}[U_\xi(\rho_S\otimes\gamma_\xi)U_\xi^{\dagger}]\\
	&=\sum_\xi p_\xi \mathcal{E}_{\xi}(\rho_{S}).
	\end{split}
	\end{equation}
Here $\mathcal{E}_{\xi}(\rho_{S})={\rm Tr}_{R_\xi}[U_\xi(\rho_S\otimes\gamma_\xi)U_\xi^{\dagger}]$ is a thermal operation induced by a heat bath with Hamiltonian $H_{\xi}$, because $\gamma_\xi$ is the thermal state of $H_\xi$, and $U_\xi$ satisfies $[U_\xi,H_{SR_\xi}]=0$. This completes the proof.
\end{proof}

Next, we will first consider the output population at point (b) in Fig. \ref{fig1}, and show a gap between the maximum values of $|\rho_{10}|$ which can be preserved by EnTO and by TO. Then we will show that this gap cannot be closed approximately, namely, the extreme state at point (b) of the EnTO cone cannot be reached by TO, even approximately.

\subsubsection{Maximum coherence via TO for output population of point (b)}
At point (b) in Fig. \ref{fig1}, the output population reads $p_0=\frac{1-q^2}{2}$ and $p_1=\frac12$. For this output population,
the transition matrix is uniquely fixed as  
	\begin{equation}
	G^{(b)}={
		\left[ \begin{array}{ccc}
		1-q^2 & 0 & 1\\
		0 & 1 & 0\\
		q^2 & 0 & 0
		\end{array} 
		\right ]}.
	\label{9}
	\end{equation}
Due to the equivalence between TO and EnTO in terms of population dynamics, this transition matrix can also be realized by TO. 
In the following, we will maximize the value of $|\rho_{10}|$ over all thermal operations which achieve the transition matrix as in Eq. (\ref{9}). The maximum value of $|\rho_{10}|$ is denoted as $|\rho_{10}^{\rm TO}|$.

From Lemma \ref{lemma1}, any thermal operation acting on our qutrit system can be written as a convex roof of thermal operations $\mathcal E_\xi$ based on heat baths with a fixed energetic gap $H_\xi=\sum_n(n\hbar\omega+\xi)\Pi_{\xi,n}$. Thus we have $G=\sum_\xi p_\xi G_\xi$, where $G_\xi$ is the transition matrix corresponding to $\mathcal E_\xi$. Meanwhile, it is easy to check that for $G^{(b)}$ as in Eq. (\ref{9}), $G^{(b)}=\sum_\xi p_\xi G_\xi$ if and only if $G_\xi=G^{(b)}$ for all $\xi$.
It means that in order to realize the transition matrix as in Eq. (\ref{9}) by TO, each $\mathcal E_\xi$ should also achieve this transition matrix. Hence, for any thermal operation $\mathcal{E}^\mathrm{TO}$ which can realize $G^{(b)}$, we have
\begin{eqnarray}
    \big|\bra{1}\mathcal{E}^\mathrm{TO}(\rho_0)\ket{0}\big| & \leq & \sum_\xi p_\xi \big|\bra{1}\mathcal{E}_\xi(\rho_0)\ket{0}\big| \nonumber\\
    & \leq & \max_{\xi}\big| \bra{1}\mathcal{E}_\xi(\rho_0)\ket{0} \big|.
\end{eqnarray}
Therefore, the maximum value of $|\rho_{10}|$ which can be achieved by TO equals to the maximum value realizable by thermal operations based on a heat bath with Hamiltonian $H_\xi$. Further, the value of $\xi$ does not affect the state transition of $S$, so we set $\xi=0$ without loss of generality. The effective Hamiltonian of the heat bath then reads $H_R=\sum_n n\hbar\omega\Pi_n$, and the Gibbs state of $R$ is
\begin{equation}
  \gamma_R=\sum_n\gamma_n\Pi_n.  \label{eq:gamma}
\end{equation}
The corresponding thermal operations are denoted as as $\mathcal{E}_0^{\rm TO}$.

The joint unitary $U$ is in the block-diagonal form $U=\bigoplus_{k=0}^\infty U^{(k)}$, where each block $U^{(k)}$ lives in a subspace with total energy $k\hbar\omega$. Precisely, $U^{(k)}$ is written as
\begin{equation}
    U^{(k)}=\sum_{i,j=0}^{\min\{2,k\}}\ket{i}\bra{j}\otimes u_{ij}^k, \label{eq:uk}
\end{equation}
where $\ket{i}$ and $\ket{j}$ are energetic states of the system $S$, $u_{ij}^k$ is a  matrix of dimension $d_{k-i}\times d_{k-j}$, and $d_n$ denotes the degeneracy of the eneragy level $E_R=n\hbar\omega$ of the heat bath.

Now we define vectors $\vec{\mathcal{U}_{ij}}$, whose $k$th entry is $(\vec{\mathcal{U}_{ij}})_k=\sqrt{\gamma_{k-j}}u_{ij}^k$ with $k=\max\{i,j\},\max\{i,j\}+1,\dots$ and $i,j=0,1,2$. For example,
\begin{eqnarray}
    \vec{\mathcal{U}_{00}} & = & (\sqrt{\gamma_0} u_{00}^0,\sqrt{\gamma_1} u_{00}^1,\dots),\\
    \vec{\mathcal{U}_{11}} & = & (\sqrt{\gamma_0} u_{11}^1,\sqrt{\gamma_1} u_{11}^2,\dots).
\end{eqnarray}
Further, for two such vectors $\vec{\mathcal A}$ and $\vec{\mathcal B}$ which satisfy that for all $l$, $\mathcal A_l$ and $\mathcal B_l$ are matrices of the same size, we define the inner product $(\vec{\mathcal A},\vec{\mathcal B})\equiv\sum_l\mathrm{Tr}(\mathcal A_l^\dagger\mathcal B_l)$. Then the inner product of vectors
$\vec{\mathcal{U}_{ij}}$ and $\vec{\mathcal{U}_{i'j'}}$ satisfying $i-i'=j-j'$ reads
\begin{equation}
    (\vec{\mathcal{U}_{i'j'}},\vec{\mathcal{U}_{ij}})=\sum_{k=\max\{i,j\}}^\infty\gamma_{k-j}\mathrm{Tr}(u_{i'j'}^{(k-j+j')\dagger}u_{ij}^k).\label{eq:inner_pro1}
\end{equation}
By substituting Eqs. (\ref{eq:gamma}) and (\ref{eq:uk}) to Eq. (\ref{1}), we obtain that for $i-i'=j-j'$,
\begin{equation}
    \bra{i}\mathcal{E}^\mathrm{TO}_0(\ket{j}\bra{j'})\ket{i'}=(\vec{\mathcal{U}_{i'j'}},\vec{\mathcal{U}_{ij}}).\label{eq:inner_pro}
\end{equation}
In particular, when $i=i'$ and $j=j'$, Eq. (\ref{eq:inner_pro}) reduces to the transition probability from $j$ to $i$
\begin{equation}
    p_{i|j}=G_{ij}=(\vec{\mathcal{U}_{ij}},\vec{\mathcal{U}_{ij}}).\label{eq:pij}
\end{equation}
Moreover, when $i=j=1$ and $i'=j'=0$, we have $\bra{1}\mathcal{E}_0^{\rm TO}(\ket{1}\bra{0})\ket{0}= (\vec{\mathcal{U}_{00}},\vec{\mathcal{U}_{11}})$. Reminding that $|\bra{1}\mathcal{E}_0^{\rm TO}(\rho_0)\ket{0}|=\frac12|\bra{1}\mathcal{E}_0^{\rm TO}(\ket{1}\bra{0})\ket{0}|$, we obtain the following
\begin{eqnarray}
    |\rho_{10}^{\rm TO}|&=&\max \frac12 |(\vec{\mathcal{U}_{00}},\vec{\mathcal{U}_{11}})| \nonumber\\
    s.t. && (\vec{\mathcal{U}_{00}},\vec{\mathcal{U}_{00}})=1-q^2,
    (\vec{\mathcal{U}_{11}},\vec{\mathcal{U}_{11}})=1.
    \label{eq:rho10_TO}
\end{eqnarray}
As we prove in Appendix \ref{app:svd}, the maximum in Eq. (\ref{eq:rho10_TO}) is reached by joint unitary operators whose main blocks $u_{jj}^k$ ($j=0,1,2$, $k\geq j$) are diagonal matrices $M_{jj}^k$ whose entries $\sigma_s(M_{jj}^k)\ (s=0,\dots,d_{k-j}-1)$ satisfy $\sigma_s(M_{jj}^k)\in[0,1]$ and are in a non-increasing order.

From Eq. (\ref{eq:pij}), if $G_{ij}=0$ for some $i$ and $j$, then $\vec{\mathcal{U}_{ij}}=0$, which means that $u_{ij}^k=\boldsymbol{0}$, $\forall k$. Therefore, the joint unitary in the thermal operation which realizes the transition matrix as in Eq. (\ref{9}) and can reach the maximum as in Eq. (\ref{eq:rho10_TO}) should be in the following form: $U^{(0)}=M_{00}^0$,
\begin{equation}
U^{(1)}=
	\left[
	\begin{array}{c|c}
	M_{00}^{1}& \textbf{0}\\ \hline 
	\textbf{0}& M_{11}^{1}
	\end{array}
	\right],\label{eq:u1}
\end{equation}
and for $k\geq2$,
\begin{equation}
U^{(k)}=
\left[
\begin{array}{c|c|c}
	M_{00}^{k}& \textbf{0}& u_{02}^{k} \\ \hline 
	\textbf{0}& M_{11}^{k}&  \textbf{0} \\ \hline
	u_{20}^{k}& \textbf{0}& \textbf{0}
\end{array}
\right].\label{eq:uk_matrix}
\end{equation}

Because each block $U^{(k)}$ is a unitary matrix, we have $M_{00}^0=\iden_{d_0}$, $M_{00}^1=\iden_{d_1}$, and $M_{11}^k=\iden_{d_{k-1}}$ for $k\geq1$. Further, for $k\geq2$, we prove in Appendix \ref{app:tilde_uk} that
\begin{equation}
    \sigma_s(M_{00}^k)=\bigg\{\begin{array}{cc}
      1,   &  s=0,\dots,d_{k}-d_{k-2}-1,\\
      0,   & s=d_k-d_{k-2},\dots,d_k-1.
    \end{array}\label{eq:sigmas}
\end{equation}
It follows that
\begin{eqnarray}
   \Tr(M_{00}^{k\dagger}M_{00}^k)&=&\Tr(M_{00}^{k\dagger}M_{11}^{k+1})\nonumber\\
   &=&\bigg\{\begin{array}{cc}
      d_{k}  &  k=0,1,\\
      d_{k}-d_{k-2}  & k\geq2.
   \end{array}  
\end{eqnarray}
It can be directly checked that this unitary can indeed achieve the transition matrix as in Eq. (\ref{9}). For example,
\begin{eqnarray}
    G_{00}& = & (\vec{\mathcal U_{00}},\vec{\mathcal U_{00}}) =\sum_{k=0}^\infty \gamma_{k}\mathrm{Tr}[M_{00}^{k\dagger}M_{00}^{k}]\nonumber\nonumber\\
    &=&\gamma_0 d_0+\gamma_1 d_1+\sum_{k=2}^\infty \gamma_{k}(d_{k}-d_{k-2})=1-q^2.
\end{eqnarray}
For the last equality, we use the fact that $\gamma_{k}=q^2\gamma_{k-2}$ and hence $\sum_{k=2}^\infty\gamma_kd_{k-2}=q^2\sum_{k=2}^\infty\gamma_{k-2}d_{k-2}=q^2$.

Now we are ready to calculate the following
\begin{eqnarray}
    |\rho_{10}^{\rm TO}| & = & \frac12(\vec{\mathcal{U}_{00}},\vec{\mathcal{U}_{11}})\nonumber\\
    & = & \frac12\sum_{k=0}^\infty \gamma_{k}\mathrm{Tr}[M_{00}^{k\dagger}M_{11}^{k+1}]\nonumber\\
    & = &\frac12 G_{00} = \frac12(1-q^2).\label{eq:to}
\end{eqnarray}
However, from Eq. (\ref{eq:max_coh}), the maximal output coherence via EnTO reads
\begin{equation}
	\vert\rho_{10}^{\rm EnTO}\vert=\frac12\sqrt{1-q^2}.
	\label{10}
\end{equation}	
Comparing Eqs. (\ref{eq:to}) and (\ref{10}), we observe the non-negligible gap between $\vert\rho_{10}^{\rm EnTO}\vert$ and $\vert\rho_{10}^{\rm TO}\vert$,
\begin{equation}
\begin{aligned}
\Delta_{10}&=\vert\rho_{10}^{\rm EnTO}\vert-\vert\rho_{10}^{\rm TO}\vert\\
&=\frac12[\sqrt{1-q^2}-(1-q^2)].
\end{aligned}
\end{equation}
It means that the states with $p_0=\frac{1-q^2}{2}$, $p_1=\frac12$ and $|\rho_{10}|\in(\frac{1-q^2}{2},\frac{\sqrt{1-q^2}}{2}]$ are in EnTO cone but not in TO cone of $\rho_0$. This extends the example in Ref. \cite{PhysRevLett.115.210403}, where only one state in $\mathcal{C}^{\rm EnTO}(\rho_0)\backslash\mathcal{C}^{\rm TO}(\rho_0)$ is found.


\subsubsection{Disproof of closure conjecture.}
Here we prove that the extreme state at point (b) of Fig. \ref{fig1} in the EnTO cone of $\rho_0$, i.e., the state with $(p_0,p_1,|\rho_{10}|)=\frac12(1-q^2,1,\sqrt{1-q^2})$, cannot be reached by TO approximately, if the temperature is proper such that $1-q-q^2\gg0$ and $q\gg0$ (or equivalently, $0\ll q\ll \frac{\sqrt{5}-1}{2}$).

Precisely, we consider states in TO cone of $\rho_0$ with $(p_0^{\epsilon_0},p_1^{\epsilon_0},|\rho_{10,\epsilon_0}|)$. Here $p_0^{\epsilon_0},p_1^{\epsilon_0}$ is a valid population distribution in the neighbourhood of point (b), i.e., $p_0^{\epsilon_0}-\frac12(1-q^2)\leq\epsilon_0$ and $\frac{1}{2}-p_1^{\epsilon_0}\leq\epsilon_0$, where $\epsilon_0$ is small. We will prove a non-negligible gap between the maximum value of $|\rho_{10,\epsilon_0}|$ and $|\rho_{10}^{\rm EnTO}|=\frac12\sqrt{1-q^2}$.

From the linearity of the population dynamics, a perturbation in the output population distribution results from a perturbation in the transition matrix. Together with Lemma \ref{lemma1}, it is sufficient to set the Gibbs state of $R$ as in Eq. (\ref{eq:gamma}) and restrict the entries of the perturbed transition matrix $G^{(b)}_\epsilon$ as $|G^{(b)}_{ij,\epsilon}-G^{(b)}_{ij}|\leq\epsilon$. Therefore, we define
\begin{eqnarray}
    |\rho^{\rm TO}_{10,\epsilon}| & = & \max \frac12 |(\vec{\mathcal{U}^\epsilon_{00}},\vec{\mathcal{U}^\epsilon_{11}})|\nonumber\\
    s.t. & & |(1-q^2)-(\vec{\mathcal{U}^\epsilon_{00}},\vec{\mathcal{U}^\epsilon_{00}})|\leq \epsilon,\nonumber\\
    & & 1-(\vec{\mathcal{U}^\epsilon_{11}},\vec{\mathcal{U}^\epsilon_{11}})\leq \epsilon.\label{eq:rho10_epsilon}
\end{eqnarray}
Here, the two restrictions come from $|G^{(b)}_{00,\epsilon}-G^{(b)}_{00}|\leq\epsilon$ and $|G^{(b)}_{11,\epsilon}-G^{(b)}_{11}|\leq\epsilon$, respectively. The following inequality is the main result of this section
\begin{eqnarray}
    \Delta_{10}^\epsilon & \equiv & |\rho_{10}^{\rm EnTO}|-|\rho_{10,\epsilon}^{\rm TO}|\nonumber\\
    & > & \frac14(1-\sqrt{1-q^2})^2(1-q-q^2)(1-\delta)-2\epsilon,\label{eq:gap_epsilon}
\end{eqnarray}
where $\delta$ is small.
It means that when the temperature is chosen properly such that $q$ and $1-q-q^2$ are not small, the output coherence $|\rho_{10}^{\rm EnTO}|=\frac12\sqrt{1-q^2}$ cannot be reached approximately by thermal operations which realizes a perturbed population dynamics.

In the following, we sketch the proof of Eq. (\ref{eq:gap_epsilon}), and leave the rigorous proof to Appendix \ref{app:gap}.

From Lemma \ref{le:optimalU} in Appendix \ref{app:svd}, the maximum in Eq. (\ref{eq:rho10_epsilon}) is reached by joint unitary operators with $u^k_{jj,\epsilon}=M^k_{jj,\epsilon}$, where $M^k_{jj,\epsilon}$ are diagonal matrices with non-negative entries in a non-increasing order. Hence $|(\vec{\mathcal{U}^\epsilon_{00}},\vec{\mathcal{U}^\epsilon_{11}})|=(\vec{\mathcal{U}^\epsilon_{00}},\vec{\mathcal{U}^\epsilon_{11}})=(\vec{\mathcal{U}^\epsilon_{11}},\vec{\mathcal{U}^\epsilon_{00}})$. Now we define
\begin{equation}
    \alpha=\frac{(\vec{\mathcal{U}^\epsilon_{00}},\vec{\mathcal{U}^\epsilon_{11}})}{(\vec{\mathcal{U}^\epsilon_{11}},\vec{\mathcal{U}^\epsilon_{11}})},\ \vec{\mu}=\vec{\mathcal{U}^\epsilon_{00}}-\alpha\vec{\mathcal{U}^\epsilon_{11}}.\label{eq:alpha_mu}
\end{equation}
Direct calculation leads to
\begin{equation}
    p_{0|0}^\epsilon p_{1|1}^\epsilon-(\vec{\mathcal{U}^\epsilon_{00}},\vec{\mathcal{U}^\epsilon_{11}})^2=p_{1|1}^\epsilon(\vec{\mu},\vec{\mu}),\label{eq:p_mu}
\end{equation}
where $p_{0|0}^\epsilon=(\vec{\mathcal{U}^\epsilon_{00}},\vec{\mathcal{U}^\epsilon_{00}})\in[1-q^2-\epsilon,1-q^2+\epsilon]$ and $p_{1|1}^\epsilon=(\vec{\mathcal{U}^\epsilon_{11}},\vec{\mathcal{U}^\epsilon_{11}})\in[1-\epsilon,1]$.
Clearly, $\sqrt{p_{0|0}^\epsilon p_{1|1}^\epsilon}\geq(\vec{\mathcal{U}^\epsilon_{00}},\vec{\mathcal{U}^\epsilon_{11}})$. Together with $p_{0|0}^\epsilon<p_{1|1}^\epsilon$, we obtain
\begin{equation}
    \frac12\sqrt{p_{0|0}^\epsilon p_{1|1}^\epsilon}-\frac12(\vec{\mathcal{U}^\epsilon_{00}},\vec{\mathcal{U}^\epsilon_{11}})>\frac14 (\vec{\mu},\vec{\mu}).\label{eq:gap_appro1}
\end{equation}
Because $\frac12\sqrt{p_{0|0}^\epsilon p_{1|1}^\epsilon}$ is $\epsilon$-close to $|\rho_{10}^{\rm EnTO}|$, and the maximum value of $\frac12(\vec{\mathcal{U}^\epsilon_{00}},\vec{\mathcal{U}^\epsilon_{11}})$ equals to $|\rho_{10,\epsilon}^{\rm TO}|$, Eq. (\ref{eq:gap_appro1}) means that the gap $\Delta_{10}^\epsilon$ cannot be closed if  $(\vec\mu,\vec\mu)$ is not small.

From the definition as in Eq. (\ref{eq:alpha_mu}) and the fact that $p_{1|1}^\epsilon\in[1-\epsilon,1]$, we obtain the following
\begin{equation}
    (\vec\mu,\vec\mu)>\sum_{k=0}^\infty\gamma_k\sum_{s=0}^{d_k-1}[\sigma_s(M_{00,\epsilon}^k)-\alpha]^2-2\alpha\epsilon.\label{eq:mumu_c}
\end{equation}
In order to evaluate the right hand side of Eq. (\ref{eq:mumu_c}), we assume the heat bath is large and satisfies the following:\\
(i) There is a set of energy levels $\mathcal R$, such that $\sum_{E_R\in\mathcal R}{\rm Tr}(\gamma_R\Pi_{E_R})=1-\delta$, where $\delta$ is small.\\
(ii) For any $E_R\in\mathcal R$, $E_R\pm m\hbar\omega\in\mathcal R$, where $m=1,2$.\\
(iii) For $E_R\in\mathcal R$, the degeneracies $g(E_R)$ satisfy $|\frac{g(E_R-m\hbar\omega)}{g(E_R)q^m}-1|\leq\delta$, where $m=1,2$.\\
These assumptions have been employed in Ref. \cite{Horodecki2013} to derive the famous thermo-majorization relation.

By assumption (iii) and the condition $1-q-q^2\gg0$,  we have
\begin{equation}
    d_k-d_{k-1}-d_{k-2}\geq [(1-q-q^2)-(q+q^2)\delta]d_k>0,
\end{equation}
for $k\hbar\omega\in\mathcal R$.
The unitarity of the block $U^{(k)}_\epsilon$ then ensures that, at least $d_k-(d_{k-1}+d_{k-2})$ singular valuses of $M_{00,\epsilon}^k$ are equal to 1. By subtracting some non-negative terms from the right hand side of Eq. (\ref{eq:mumu_c}), we obtain the following
\begin{eqnarray}
     (\vec\mu,\vec\mu)& > & \sum_{k:k\hbar\omega\in\mathcal R}\gamma_k\sum_{s:\sigma_s(M_{00,\epsilon}^k)=1}[\sigma_s(M_{00,\epsilon}^k)-\alpha]^2-2\alpha\epsilon\nonumber\\
     &=&\sum_{k:k\hbar\omega\in\mathcal R}\gamma_k(d_k-d_{k-1}-d_{k-2})(1-\alpha)^2-2\alpha\epsilon\nonumber\\
     &\geq&(1-\alpha)^2(1-q-q^2)(1-\delta)-2\alpha\epsilon.\label{eq:mumu1_c}
\end{eqnarray}
Here $\alpha\leq\sqrt{\frac{p_{0|0}^\epsilon}{p_{1|1}^\epsilon}}\leq \sqrt{\frac{1-q^2+\epsilon}{1-\epsilon}}\equiv\alpha(\epsilon)$ by definition. Clearly, $\alpha(\epsilon)$ is $\epsilon$-close to $\sqrt{1-q^2}$. Together with Eq. (\ref{eq:gap_appro1}), we arrive at Eq. (\ref{eq:gap_epsilon}).

It is worth mentioning that, although our derivation heavily depends on the condition $1-q-q^2\gg0$, this condition is by no means necessary for the gap. Moreover, the bound to the gap as in Eq. (\ref{eq:gap_epsilon}) is not tight. Nevertheless, these results are sufficient for the the purpose of this paper, which is to disprove the closure conjecture with a counterexample. We will leave the explicit problems, such as necessary conditions and tighter bounds on the gap, to future work.

\section{Conclusions}
We have disproved the closure conjecture with an analytic counterexample. We derive the EnTO cone of a given qutrit state, and calculate the optimal coherence preserved by TO for a given output population distribution. We also evaluate the upper bound on the coherence preserved by TO, if a small perturbation on the output population distribution is allowed. By doing so, we discover a state conversion under EnTO which cannot be approximated by TO.

Our findings show that thermal operations and enhanced thermal operations can lead to different laws of coherence evolution. Further, the methods we developed here can be used to evaluate the upper bound on the output coherence via thermal operations. Thus our results can contribute to studying the restrictions on coherence dynamics under TO.

\begin{acknowledgments}
This work was supported by National Natural Science Foundation of China under Grant No. 11774205, and the Young Scholars Program of Shandong University.
\end{acknowledgments}

\appendix

\section{Analytic expression for EnTO cone of $\rho_0$}\label{app:ento_cone}
In this section, we analytically solve the optimization problem in Eq. (\ref{eq:max_coh}).
By observing the objective function in Eq. (\ref{eq:max_coh}), we notice that the maximum can be taken at the point $(G_{00},G_{11})$ where both $G_{00}$ and $G_{11}$ are maximal in the feasible region. 

Here we first derive the feasible region of optimization. The transition conditions $G\boldsymbol{\gamma}=\boldsymbol{\gamma}, G\boldsymbol{p_0}=\boldsymbol{p}$ give us following equations
\begin{eqnarray}
        G_{01}&=&2p_0-G_{00},\label{eq:A1}\\
        G_{10}&=&2p_1-G_{11},\\
        G_{02}&=&[1-G_{00}-q(2p_0-G_{00})]/q^2,\\
        G_{12}&=&[1-qG_{11}-(2p_1-G_{11})]/q^2,\label{eq:A4}
\end{eqnarray}
where $(p_0,p_1)$ is a fixed point in the hexogen in Fig. \ref{fig1}.
By applying the condition that above entries lie in $[0,1]$, we give the bounds of $G_{00}$ and $G_{11}$ as
\begin{eqnarray}
        G_{00}^{\rm lb}\leq G_{00}\leq G^{\rm ub}_{00}&:=&\min\left\{2p_0,1,\frac{1-2qp_0}{1-q}\right\},\label{g00upper}\\
        G_{11}^{\rm lb}\leq G_{11}\leq G^{\rm ub}_{11}&:=&\min\left\{2p_1,1,\frac{q^2+2p_1-1}{1-q}\right\}.\label{g11upper}
\end{eqnarray}
Here we omit the expressions for the lower bounds $G_{00}^{\rm lb}$ and $G_{11}^{\rm lb}$, because the central question here is to find the upper bound for $\sqrt{G_{00}G_{11}}$.
Further, by applying the condition $G^{\mathrm{T}}\mathbb{I}=\mathbb{I}$,
we get that $G_{00}+G_{10}$, $G_{01}+G_{11}$ and $G_{02}+G_{12}$ also lie in $[0,1]$. 
It follows that
\begin{equation}
    \mathrm{lb}\leq G_{00}-G_{11}\leq\mathrm{ub},\label{eq:g00_g11}
\end{equation}
where 
\begin{eqnarray}
    \mathrm{lb}&:=&\max\left\{-2p_1,2p_0-1,\frac{1-2(qp_0+p_1)}{1-q}+q\right\},\\
    \mathrm{ub}&:=&\min\left\{1-2p_1,2p_0,\frac{1+q-2(qp_0+p_1)}{1-q}\right\}.
\end{eqnarray}
The combination of Eqs. (\ref{g00upper}), (\ref{g11upper}) and (\ref{eq:g00_g11}) gives the necessary and sufficient condition for entries $G_{00}$ and $G_{11}$ in a feasible transition matrix $G$. The reason for sufficiency is as follows. If one starts from a given pair of $G_{00}$ and $G_{11}$ which satisfies these three equations, other entries of $G$ are fixed by Eqs. (\ref{eq:A1}-\ref{eq:A4}) and $G^{\mathrm{T}}\mathbb{I}=\mathbb{I}$. Further, the transition matrix $G$ as such satisfies all of the four conditions in Eq. (\ref{eq:max_coh}).

Next, we calculate the maximum in Eq. (\ref{eq:max_coh}) for the following three cases
\begin{itemize}
    \item {\rm Case 1.} $\mathrm{lb}\leq G^{\rm ub}_{00}-G^{\rm ub}_{11}\leq\mathrm{ub}$,
    \item {\rm Case 2.} $\mathrm{lb}> G^{\rm ub}_{00}-G^{\rm ub}_{11}$,
    \item {\rm Case 3.} $G^{\rm ub}_{00}-G^{\rm ub}_{11}>\mathrm{ub}$.
\end{itemize}
For Case 1, it is easy to see that $G^{\star}_{00}=G^{\rm ub}_{00}$ and $G^{\star}_{11}=G^{\rm ub}_{11}$. For Case 2, the upper bound $G^{\rm ub}_{11}$ cannot be reached by $G^{\star}_{11}$ because of Eq. (\ref{eq:g00_g11}), so we have $G^{\star}_{00}=G^{\rm ub}_{00}$ and $G^{\star}_{11}=G^{\rm ub}_{00}-\mathrm{lb}$. Similarly, for Case 3, we have $G^{\star}_{00}=G^{\rm ub}_{11}+\mathrm{ub}$ and $G^{\star}_{11}=G^{\rm ub}_{11}$. To sum up, we arrive at the following analytic solution to Eq. (\ref{eq:max_coh})
\begin{equation}
    \frac{1}{2}\left\{\begin{aligned}
    &\sqrt{G^{\rm ub}_{00}G^{\rm ub}_{11}},&\mathrm{lb}\leq &G^{\rm ub}_{00}-G^{\rm ub}_{11}\leq\mathrm{ub};\\
    &\sqrt{G^{\rm ub}_{00}(G^{\rm ub}_{00}-\mathrm{lb})},&\mathrm{lb}> &G^{\rm ub}_{00}-G^{\rm ub}_{11};\\
    &\sqrt{(G^{\rm ub}_{11}+\mathrm{ub})G^{\rm ub}_{11}},&&G^{\rm ub}_{00}-G^{\rm ub}_{11}>\mathrm{ub}.
    \end{aligned}\right.
\end{equation}
This solution is visualized in Fig. \ref{fig1}.

\section{Optimal joint unitary}\label{app:svd}
In this section, we will show that, among all of the joint unitary operators which achieve a given transition matrix, the maximum value of $|\rho_{10}|$ is reached by the unitary operators whose main blocks are diagonal with non-negative entries in a non-increasing order, give. Precisely, we will prove the following lemma.

\begin{lemma}\label{le:optimalU}
  Let $H_S=\sum_{m=0}^2m\hbar\omega\ket{m}\bra{m}$ and $H_R=\sum_n n\hbar\omega\Pi_n$ be the Hamiltonian of the main system $S$ and that of the heat bath $R$ respectively, and $\ket{\psi_0}=\frac{1}{\sqrt{2}}(\ket{0}+\ket{1})$ be the initial state of $S$. For any (energy-preserving) joint unitary $U$, one can always find another (energy-preserving) joint unitary $\tilde U$, which satisfies the following:\\
(a) $\tilde{u}_{jj}^k=M_{jj}^k$, where $M_{jj}^k$ are diagonal matrices with non-negative entries in a non-increasing order;\\
(b) $(\vec{\mathcal{U}_{ij}},\vec{\mathcal{U}_{ij}})=(\vec{\widetilde{\mathcal{U}}_{ij}},\vec{\widetilde{\mathcal{U}}_{ij}})$, which means that $U$ and $\widetilde{U}$ lead to the same transition matrix;\\
(c) $|(\vec{\mathcal{U}_{00}},\vec{\mathcal{U}_{11}})|\leq |(\vec{\widetilde{\mathcal{U}}_{00}},\vec{\widetilde{\mathcal{U}}_{11}})|$, which means that the output coherence $|\rho_{10}|$ achieved by $\widetilde{U}$ is no less than that achieved by $U$.  
\end{lemma}

The rest of this section will be devoted to the proof of this lemma.
Following the method as in Ref. \cite{PhysRevLett.115.210403}, we perform the singular value decomposition (SVD) to the main blocks 
\begin{equation}
    u_{jj}^k=A_{jj}^k M_{jj}^k B_{jj}^{k\dagger}, \label{eq:svd}
\end{equation}
where $A_{jj}^k$ and $B_{jj}^k$ are unitary matrices, and $M_{jj}^k$ are diagonal matrices with singular values (which are non-negative numbers) in a non-increasing order as entries. By introducing two unitary matrices $A\equiv \bigoplus_{k=0}^\infty[\bigoplus_{j=0}^{\min\{k,2\}}A_{jj}^k]$ and $B\equiv \bigoplus_{k=0}^\infty[\bigoplus_{j=0}^{\min\{k,2\}}B_{jj}^k]$, we define a new joint unitary
\begin{equation}
    \widetilde{U}=A^\dagger U B.\label{eq:u_tilde}
\end{equation}
Similar to Eq. (\ref{eq:u}), $\widetilde{U}$ can be expressed as $\widetilde{U}=\bigoplus_{k=0}^\infty\widetilde{U}^{(k)}$ with $\widetilde{U}^{(k)}=\sum_{i,j=0}^{\min\{k,2\}}\ket{i}\bra{j}\otimes \widetilde{u}_{ij}^k$. Then Eq. (\ref{eq:u_tilde}) leads to
\begin{equation}
    \widetilde{u}_{ij}^k=A_{ii}^{k\dagger}u_{ij}^k B_{jj}^k.\label{eq:u_ij_tilde}
\end{equation}
Now we are ready to prove that $\widetilde{U}$ satisfies the conditions (a), (b) and (c) as mentioned above.

By substituting Eq. (\ref{eq:svd}) to Eq. (\ref{eq:u_ij_tilde}) with $i=j$, we have $\widetilde{u}_{jj}^k=A_{jj}^{k\dagger}A_{jj}^k M_{jj}^k B_{jj}^{k\dagger} B_{jj}^k=M_{jj}^k$. Thus, condition (a) holds.

As for condition (b), we employ the definition of inner product as in Eq. (\ref{eq:inner_pro1}), and obtain the following
\begin{eqnarray}
    &&(\vec{\widetilde{\mathcal{U}}_{ij}},\vec{\widetilde{\mathcal{U}}_{ij}})\nonumber\\
    &=&\sum_{k=\max\{i,j\}}^\infty\gamma_{k-j}{\rm Tr}(\widetilde{u}_{ij}^{k\dagger}\widetilde{u}_{ij}^k)\nonumber\\
    &=&\sum_{k=\max\{i,j\}}^\infty\gamma_{k-j}{\rm Tr}( B_{jj}^{k\dagger} u_{ij}^{k\dagger}A_{ii}^kA_{ii}^{k\dagger}u_{ij}^k B_{jj}^k)\nonumber\\
    &=&\sum_{k=\max\{i,j\}}^\infty\gamma_{k-j}{\rm Tr}( u_{ij}^{k\dagger}u_{ij}^k)\nonumber\\
    &=&(\vec{\mathcal{U}_{ij}},\vec{\mathcal{U}_{ij}}).
\end{eqnarray}
Then from Eq. (\ref{eq:pij}), this equality means that $U$ and $\widetilde{U}$ leads to the same transition matrix.

In order to prove condition (c), we employ the following lemma, which was proved in Refs. \cite{von1937some,Fan1951} and introduced to this problem in Ref. \cite{PhysRevLett.115.210403}.
\begin{lemma} If $X$ and $Y$ are $d \times d$ complex matrices, $W$ and $V$ are $d \times d$ unitary  matrices, and $\sigma_{1}\geq\sigma_{2}\geq\cdots\sigma_{d}\geq0$ denotes ordered singular values, then 	$$\left|{\rm Tr}\left(WXVY\right)\right.|\leq\sum_{s=0}^{d-1}\sigma_{s}(X)\sigma_{s}(Y)$$
	and the equality always exist for some $W$ and $V$.
\end{lemma}
From this lemma, we have $|\mathrm{Tr}[u_{00}^{(k-1)\dagger}u_{11}^{k}]|\leq \sum_{s=0}^{d_{k-1}-1}\sigma_s(u_{00}^{(k-1)\dagger})\sigma_s(u_{11}^{k})=\mathrm{Tr}[M_{00}^{(k-1)\dagger}M_{11}^{k}]$, and hence,
\begin{eqnarray}
    |(\vec{\mathcal{U}_{00}},\vec{\mathcal{U}_{11}})| & \leq & \sum_{k=1}^\infty \gamma_{k-1}|\mathrm{Tr}[u_{00}^{k\dagger}u_{11}^{k-1}]|\nonumber\\
    & \leq & \sum_{k=1}^\infty \gamma_{k-1}\mathrm{Tr}[M_{00}^{(k-1)\dagger}M_{11}^{k}]\nonumber\\
    & = & |(\vec{\widetilde{\mathcal{U}}_{00}},\vec{\widetilde{\mathcal{U}}_{11}})|.
\end{eqnarray}
This completes the proof of Lemma \ref{le:optimalU}.

\section{Proof of Eq. (\ref{eq:sigmas})}\label{app:tilde_uk}
In this section, we will prove Eq. (\ref{eq:sigmas}) in the main text.
For $k\geq 2$, we have 
\begin{equation}
U^{(k)}=
\left[
\begin{array}{c|c|c}
	M_{00}^{k}& \textbf{0} &  u_{02}^{k} \\ \hline 
	\textbf{0}& M_{11}^{k} &  \textbf{0} \\ \hline
	 u_{20}^{k}& \textbf{0} & \textbf{0}
\end{array}
\right],\label{eq:tilde_uk}
\end{equation}
where $M_{11}^k=\iden_{d_{k-1}}$ and $M_{00}^k$ is a diagonal matrix with non-negative entries in a non-increasing order. We will first prove that the number of zero rows in ${u}_{02}^k$ equals exactly to $d_k-d_{k-2}$, and then show that there are $d_k-d_{k-2}$ entries of $M_{00}^k$ which equal to one, while the rest entries of $M_{00}^k$ all equal to zero.

Here and following, we denote the $s$th row of a matrix $X$ as $\vec{r}_s(X)$, and the $l$th column of $X$ as $\vec{c}_l(X)$. From the unitarity of $U^{(k)}$, we have $\vec{r}_s( U^{(k)})^*\cdot\vec{r}_{s'}( U^{(k)})=\delta_{ss'}$ and $\vec{c}_l( U^{(k)})^*\cdot\vec{c}_{l'}( U^{(k)})=\delta_{ll'}$. Our deviations are based on these two equations.

For $d_k+d_{k-1}\leq l,l'<d_k+d_{k-1}+d_{k-2}$, we have $\delta_{ll'}=\vec{c}_l(U^{(k)})^*\cdot\vec{c}_{l'}(U^{(k)})=\vec{c}_l(u_{02}^k)^*\cdot\vec{c}_{l'}(u_{02}^k)$. It means that the $d_{k-2}$ columns of $u_{02}^k$ are nontrivial and linearly independent. Therefore, at least $d_{k-2}$ rows of $u_{02}^k$ are nontrivial.

For $0\leq s,s'\leq d_k-1$ and $s\neq s'$, we have $0=\vec{r}_s(U^{(k)})^*\cdot\vec{r}_{s'}( U^{(k)})=\vec{r}_s(u_{02}^k)^*\cdot\vec{r}_{s'}( u_{02}^k)$. Hence, there are at most $d_{k-2}$ nontrivial rows in $ u_{02}^k$, and the rest rows of $u_{02}^k$ have to be zero.

Therefore, the number of zero rows in ${u}_{02}^k$ is $d_k-d_{k-2}$.
Then we have
\begin{equation}
    \vec{r}_s(u_{02}^k)\bigg\{\begin{array}{cc}
      =0,   &  0\leq s\leq d_k-d_{k-2}-1,\\
      \neq 0,   & d_k-d_{k-2}\leq s\leq d_k-1.
    \end{array}
\end{equation}
It follows that for $0\leq s\leq d_k-d_{k-2}-1$, we have $1=\vec{r}_s(U^{(k)})^*\cdot\vec{r}_{s}(U^{(k)})=[\sigma_s(M_{00}^k)]^2$, and hence $\sigma_s(M_{00}^k)=1$.

For $d_k-d_{k-2}\leq s\leq d_k-1$ and $d_k+d_{k+1}\leq l\leq d_k+d_{k+1}+d_{k+2}$, the scalar product $\vec{c}_s( U^{(k)})^*\cdot\vec{c}_{l}( U^{(k)})=0$ gives $\sigma_s(M_{00}^k)\cdot \vec{r}_s(u_{02}^k)=0$. Meanwhile, $\vec{r}_s( u_{02}^k)\neq0$, so we have $\sigma_s(M_{00}^k)=0$.

\section{Detailed Proof of Eq. (\ref{eq:gap_epsilon})}\label{app:gap}
In this section, we prove Eq. (\ref{eq:gap_epsilon}) under the condition $1-q-q^2\gg0$ (or equivalently, $q\ll \frac{\sqrt5-1}{2}$) and the assumptions (i)-(iii) on the heat bath.

Here we start from Eq. (\ref{eq:p_mu}) in the main text.
It follows that $\sqrt{p_{0|0}^\epsilon p_{1|1}^\epsilon}\geq(\vec{\mathcal{U}^\epsilon_{00}},\vec{\mathcal{U}^\epsilon_{11}})$ and hence,
\begin{equation}
    \frac12\sqrt{p_{0|0}^\epsilon p_{1|1}^\epsilon}-\frac12(\vec{\mathcal{U}^\epsilon_{00}},\vec{\mathcal{U}^\epsilon_{11}})\geq\frac14 \sqrt{\frac{p_{1|1}^\epsilon}{p_{0|0}^\epsilon}}(\vec{\mu},\vec{\mu}),\label{eq:p_mu1}
\end{equation}
Because $\sqrt{p_{0|0}^\epsilon p_{1|1}^\epsilon}\leq\sqrt{1-q^2+\epsilon}=\sqrt{1-q^2}+\frac{\epsilon}{2\sqrt{1-q^2}}+O(\epsilon^2)$, and $\sqrt{\frac{p_{1|1}^\epsilon}{p_{0|0}^\epsilon}}\geq\sqrt{\frac{1-\epsilon}{1-q^2+\epsilon}}=\frac{1}{\sqrt{1-q^2}}-\frac{2-q^2}{2(1-q^2)^{3/2}}\epsilon+O(\epsilon^2)$, the above equation becomes
 
\begin{eqnarray}
    \Delta_{10}^\epsilon&\geq&\frac14\left(1-\frac{2-q^2}{2(1-q^2)}\epsilon\right) \frac{(\vec{\mu},\vec{\mu})}{\sqrt{1-q^2}}\nonumber\\
    &&-\frac{\epsilon}{4\sqrt{1-q^2}}+O(\epsilon^2), \label{eq:delta_epsilon1}
\end{eqnarray}
From the definition as in Eq. (\ref{eq:alpha_mu}) and the fact that $p_{1|1}^\epsilon\in[1-\epsilon,1]$, we obtain the following
\begin{equation}
    (\vec\mu,\vec\mu)>\sum_{k=0}^\infty\gamma_k\sum_{s=0}^{d_k-1}[\sigma_s(M_{00,\epsilon}^k)-\alpha]^2-2\alpha\epsilon.\label{eq:mumu}
\end{equation}
The reason is as follows. Firstly, because $1-p_{1|1}^\epsilon\leq\epsilon$, $\sum_{k=0}^\infty\gamma_k\sum_{s=0}^{d_k-1}1=1$ and $p_{1|1}^\epsilon=\sum_{k=0}^\infty\gamma_k\sum_{s=0}^{d_k-1}[\sigma_s(M_{11,\epsilon}^k)]^2$, we have the following
\begin{equation}
    \sum_{k=0}^\infty\gamma_k\sum_{s=0}^{d_k-1}[1-\sigma_s(M_{11,\epsilon}^k)]\leq \sum_{k=0}^\infty\gamma_k\sum_{s=0}^{d_k-1}[1-\sigma_s(M_{11,\epsilon}^k)^2]\leq\epsilon.
\end{equation}
By the definition as in Eq. (\ref{eq:alpha_mu}),
\begin{eqnarray}
     (\vec\mu,\vec\mu)&=&\sum_{k=0}^\infty\gamma_k\sum_{s=0}^{d_k-1}[\sigma_s(M_{00,\epsilon}^k)-\alpha\sigma_s(M_{11,\epsilon}^k)]^2\nonumber\\
     &\geq&\sum_{k=0}^\infty\gamma_k\sum_{s=0}^{d_k-1}\{[\sigma_s(M_{00,\epsilon}^k)-\alpha]^2-2\alpha[1-\sigma_s(M_{11,\epsilon}^k)]\}\nonumber\\
     &\geq&\sum_{k=0}^\infty\gamma_k\sum_{s=0}^{d_k-1}[\sigma_s(M_{00,\epsilon}^k)-\alpha]^2-2\alpha\epsilon.
\end{eqnarray}
Here for the second line, we have used $\sigma_s(M_{11,\epsilon}^k)\in[0,1]$ and $\sigma_s(M_{00,\epsilon}^k)-\alpha\geq-1$.

Now we consider the $k$th block of $U^\epsilon$
\begin{equation}
    U_\epsilon^{(k)}=\left[\begin{array}{c|c|c}
        M_{00,\epsilon}^k & u_{01,\epsilon}^k & u_{02,\epsilon}^k \\ \hline
        u_{10,\epsilon}^k & M_{11,\epsilon}^k & u_{12,\epsilon}^k \\ \hline
        u_{20,\epsilon}^k & u_{21,\epsilon}^k & u_{22,\epsilon}^k
    \end{array}
    \right],
\end{equation}
where $k\hbar\omega\in\mathcal R$. By assumption (iii) and the condition $1-q-q^2\gg0$, we have $d_k>d_{k-1}+d_{k-2}$.
For $0\leq s,s'\leq d_{k}-1$ and $s\neq s'$, the unitarity of $U^{(k)}_\epsilon$ implies 
\begin{eqnarray}
    0&=&\vec r_s(U_\epsilon^{(k)})^*\cdot \vec r_{s'}(U_\epsilon^{(k)})\nonumber\\
    &=& \vec r_s([u_{01,\epsilon}^k | u_{02,\epsilon}^k])^*\cdot \vec r_{s'}([u_{01,\epsilon}^k | u_{02,\epsilon}^k]).\label{eq:inner_pro_epsilon}
\end{eqnarray}
Here $[u_{01,\epsilon}^k | u_{02,\epsilon}^k]$ is a rectangle matrix of dimension $d_k\times (d_{k-1}+d_{k-2})$, which is composed of two blocks $u_{01,\epsilon}^k$ and $u_{02,\epsilon}^k$. Hence there are at most $(d_{k-1}+d_{k-2})$ nontrivial rows of $[u_{01,\epsilon}^k | u_{02,\epsilon}^k]$ satisfying Eq. (\ref{eq:inner_pro_epsilon}), and the rest $d_k-(d_{k-1}+d_{k-2})$ rows are zero. For $\vec r_s([u_{01,\epsilon}^k | u_{02,\epsilon}^k])=0$, the equality $\vec r_s(U_\epsilon^{(k)})^*\cdot \vec r_{s}(U_\epsilon^{(k)})=1$ leads to $\sigma_s(M_{00,\epsilon}^k)=1$. Therefore, at least $d_k-(d_{k-1}+d_{k-2})$ singular values of $M_{00,\epsilon}^k$ are equal to 1. By subtracting some non-negative terms from the right hand side of Eq. (\ref{eq:mumu}), we obtain the following
\begin{eqnarray}
     (\vec\mu,\vec\mu)& > & \sum_{k:k\hbar\omega\in\mathcal R}\gamma_k\sum_{s:\sigma_s(M_{00,\epsilon}^k)=1}[\sigma_s(M_{00,\epsilon}^k)-\alpha]^2-2\alpha\epsilon\nonumber\\
     &=&\sum_{k:k\hbar\omega\in\mathcal R}\gamma_k(d_k-d_{k-1}-d_{k-2})(1-\alpha)^2-2\alpha\epsilon\nonumber\\
     &\geq&(1-\alpha)^2(1-q-q^2)(1-\delta)-2\alpha\epsilon.\label{eq:mumu1_c}
\end{eqnarray}
For the last line, we have used $\gamma_k=\gamma_{k-m}q^m$ with $m=0,1,2$, and hence, $\sum_{k:k\hbar\omega\in\mathcal R}\gamma_kd_{k-m}=q^m\sum_{k:k\hbar\omega\in\mathcal R}\gamma_{k-m}d_{k-m}=q^m(1-\delta)$.
By definition, $\alpha\leq\sqrt{p_{0|0}^\epsilon/p_{1|1}^\epsilon}\leq \sqrt{\frac{1-q^2+\epsilon}{1-\epsilon}}\equiv\alpha(\epsilon)$, and hence $\alpha\leq\sqrt{1-q^2}+O(\epsilon)$, and
\begin{widetext}
\begin{eqnarray}
     (1-\alpha)^2 &\geq& [1-\alpha(\epsilon)]^2= [1-\alpha(\epsilon)]^2|_{\epsilon=0}-2[1-\alpha(\epsilon)]\frac{d\alpha(\epsilon)}{d\epsilon}\bigg|_{\epsilon=0}\epsilon+O(\epsilon^2)\nonumber\\
     &=& [1-\sqrt{1-q^2}]^2-\frac{(1-\sqrt{1-q^2})(2-q^2)}{\sqrt{1-q^2}}\epsilon+O(\epsilon^2).\label{eq:alpha1}
\end{eqnarray}
\end{widetext}
Then Eq. (\ref{eq:mumu1_c}) becomes
\begin{equation}
         (\vec\mu,\vec\mu)>
     [1-\sqrt{1-q^2}]^2(1-q-q^2)(1-\delta)-f_1(q)\epsilon+O(\epsilon^2).\label{eq:mumu2}
\end{equation}
where $f_1(q)=2\sqrt{1-q^2}+\frac{(1-\sqrt{1-q^2})(2-q^2)(1-q-q^2)}{\sqrt{1-q^2}}$. Together with Eq. (\ref{eq:delta_epsilon1}), we obtain
\begin{eqnarray}
    \Delta_{10}^\epsilon&>&\frac{[1-\sqrt{1-q^2}]^2(1-q-q^2)(1-\delta)}{4\sqrt{1-q^2}}\nonumber\\
    &&-\frac{1}{4}f(q)\epsilon+O(\epsilon^2), \label{eq:delta_epsilon2}
\end{eqnarray}
where $f(q)=2+\frac{[1-(1-q^2)^2](1-q-q^2)}{2(1-q^2)^{3/2}}+\frac{1}{\sqrt{1-q^2}}$. For $0< q<\frac{\sqrt5-1}{2}$, it holds that $1<\frac{1}{\sqrt{1-q^2}}<2$ and $f(q)<8$. Then we obtain the simpler form as in Eq. (\ref{eq:gap_epsilon}).

\nocite{*}

\bibliography{gap}

\end{document}